\documentclass[11pt]{article}
\usepackage[margin=1in,bottom=1in]{geometry}
\usepackage[utf8]{inputenc}

\usepackage{libertine}
\usepackage{amsthm}
\usepackage{amsmath}

\usepackage{amssymb}
\usepackage{stmaryrd}
\usepackage{mathrsfs} 
\usepackage{mathtools}
\usepackage{thm-restate}
\usepackage[colorlinks = true,linkcolor = blue,urlcolor = blue, citecolor = blue]{hyperref}
\usepackage[capitalize, nameinlink]{cleveref}
\usepackage{url}
\usepackage{todonotes}
\usepackage{mathrsfs} 
\usepackage{soul}
\usepackage{relsize}
\usepackage{enumerate}
\usepackage[all,defaultlines=3]{nowidow}
\usepackage[mathlines]{lineno}
\usepackage{multicol}
\usepackage{enumitem}
\usepackage{xspace}
\usepackage{lineno}
\usepackage{nicefrac}
\usepackage{wasysym}

\usepackage{textpos}
\setlist[itemize]{topsep=4pt,itemsep=3pt,parsep=0pt} 
\setlist[enumerate]{topsep=4pt,itemsep=3pt,parsep=0pt} 

\crefname{claim}{Claim}{Claims}
\crefname{figure}{Figure}{Figures}
\renewcommand{\preceq}{\preccurlyeq}

\newtheorem{theorem}{Theorem}[section]

\newtheorem{corollary}[theorem]{Corollary}

\newtheorem{lemma}[theorem]{Lemma}

\newtheorem{claim}[theorem]{Claim}

\theoremstyle{definition}
\newtheorem{definition}[theorem]{Definition}
\theoremstyle{plain}

\theoremstyle{definition}

\newtheorem*{example*}{Example}

\numberwithin{equation}{section}

\newenvironment{claimproof}[1][Proof of the claim.]{%
  \begin{proof}[#1]%
}{%
  \end{proof}%
}

\def\phi{\varphi}
\newcommand{\N}{\mathbb{N}}

\newcommand{\Tt}{\mathscr{T}}
\newcommand{\Hh}{\mathscr{H}}

\newcommand{\Ss}{\mathcal{S}}

\newcommand{\aff}[1]{\textcolor{black!60}{\small{#1}}}

\newcommand{\Cols}{\mathcal{C}}
\newcommand{\Tf}{\mathsf{T}}

\newcommand{\Minors}{\mathsf{Minors}}
\newcommand{\WReach}{\mathsf{WReach}}
\newcommand{\wcol}{\mathsf{wcol}}

\newcommand{\Cc}{\mathscr C}
\newcommand{\Dd}{\mathscr D}
\def\epsilon{\varepsilon}

\newcommand{\Oh}{\mathcal{O}}

\renewcommand{\emptyset}{\varnothing}

\renewcommand{\leq}{\leqslant}
\renewcommand{\geq}{\geqslant}
\renewcommand{\le}{\leq}
\renewcommand{\ge}{\geq}

\renewcommand{\setminus}{-}

\newcommand{\Ramsey}{\mathsf{Ramsey}}

\usepackage{tikz}

\usetikzlibrary{backgrounds}
\tikzset{node/.style={draw, circle, fill = black, minimum size = 3pt, inner sep=0pt, line width=1pt}}
\tikzset{nodesubwall/.style={draw, circle, fill = blue, minimum size = 3.1pt, inner sep=0pt}}
\tikzset{corner/.style={draw=magenta, fill = red!20!white, minimum size = 6pt,inner sep=0pt,line width=1pt}}
\tikzset{central/.style={draw=orange, fill = red!20!white, minimum size = 6pt,inner sep=0pt,line width=1pt}}
\tikzset{edge/.style={draw=white!60!black,line width=1.5pt}}
\tikzset{edgesubwall/.style={draw=blue!60!white,line width=2pt}}

\tikzset{subnode/.style={draw, circle, fill = yellow!50!red!50!white, minimum size = 2pt, inner sep=0pt}}
\tikzset{edgethin/.style={draw=white!60!black,line width=1.3pt}}

\begin{document}

\newcommand{\funding}{M.P. was supported by the project BOBR that is funded from the European Research Council (ERC) under the European Union’s Horizon 2020 research and innovation programme with grant agreement No. 948057.
S.T. received funding from ERC grant BUKA (No. 101126229). J.G. and J.G. are supported by the Polish National Science Centre
SONATA-18 grant number 2022/47/D/ST6/03421.
}

\title{First-order transducibility among classes of sparse graphs\footnote{\funding}}
\date{}
 \author{
   Jakub Gajarsk\'y \\
   \aff{University of Warsaw} \\
   \aff{Masaryk University, Brno} \\
   \aff{gajarsky@fi.muni.cz}
   \and
   Jeremi Gładkowski \\
   \aff{University of Warsaw} \\
   \aff{Masaryk University, Brno} \\
   \aff{j.gladkowski@mimuw.edu.pl}
   \and
   Jan Jedelsk\'y \\
   \aff{Masaryk University, Brno} \\
   \aff{xjedelsk@fi.muni.cz}
   \and
   Michał Pilipczuk \\
   \aff{University of Warsaw} \\
   \aff{michal.pilipczuk@mimuw.edu.pl}
   \and
   Szymon Toruńczyk \\
   \aff{University of Warsaw} \\
   \aff{szymtor@mimuw.edu.pl}
 }
\maketitle
\thispagestyle{empty}

\begin{abstract}
	\noindent We prove several negative results about first-order transducibility for classes of sparse graphs:
	\begin{itemize}
		\item for every $t\in \N$, the class of graphs of treewidth at most $t+1$ is not transducible from the class of graphs of treewidth at most $t$;
		\item for every $t\in \N$, the class of graphs with Hadwiger number at most $t+2$ is not transducible from the class of graphs with Hadwiger number at most $t$; and
		\item the class of graphs of treewidth at most $4$ is not transducible from the class of planar graphs.
	\end{itemize}
	These results are obtained by combining the known upper and lower bounds on the weak coloring numbers of the considered graph classes with the following two new observations:
	\begin{itemize}
		\item If a weakly sparse (that is, $K_{t,t}$-free, for some $t \in \N$) graph class $\Dd$ is transducible from a class $\Cc$ of bounded expansion, then for some $k\in \N$, every graph $G\in \Dd$ is a $k$-congested depth-$k$ minor of a graph $H^\circ$ obtained from some $H\in \Cc$ by adding a universal vertex.
		\item The operations of adding a universal vertex and of taking $k$-congested depth-$k$ minors, for a fixed $k$, preserve the degree of the distance-$d$ weak coloring number of a graph class, understood as a polynomial in $d$.
	\end{itemize}
\end{abstract}

 \begin{textblock}{20}(-1.75, 3.6)
 \includegraphics[width=40px]{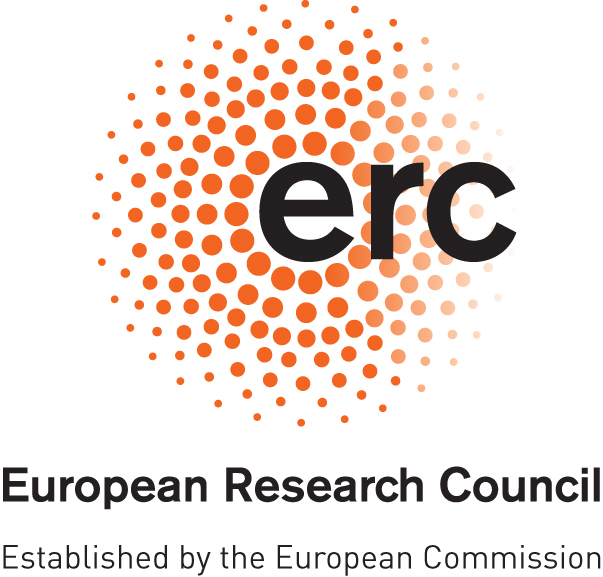}%
 \end{textblock}
 \begin{textblock}{20}(-1.75, 4.6)
 \includegraphics[width=40px]{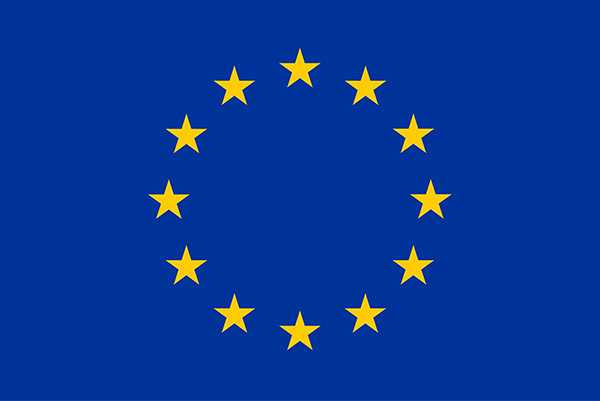}%
 \end{textblock}

\newpage

\clearpage
\setcounter{page}{1}
 
\section{Introduction}

{\em{Transductions}} are a basic model-theoretic notion whose aim is to capture the concept of encoding one graph in another graph by means of logical formulas. A transduction $\Tf$ is specified by a set of colors $\Cols$ and a symmetric formula $\phi(x,y)$ in the signature of graphs with vertices colored with $\Cols$.
We say that a graph $H$ can be {\em{transduced}} from a graph $G$ using $\Tf$ if the adjacency relation of $H$ can be interpreted in some $\Cols$-coloring $G^+$ of $G$ using $\phi$ (see \cref{sec:prelims} for a formal definition). This concept can be applied to different logics, yielding notions of transductions of different power. In this paper we focus on first-order transductions, where $\phi$ is required to be a formula of first-order logic on ($\Cols$-colored) graphs. (For brevity, we will henceforth drop the prefix {\em{first-order}}, assuming all transductions to be such.) 
Transductions are an abstract notion that captures many commonly considered graph operations of both local and global nature, such as taking the power of a graph, taking an induced subgraph, or taking a complement.

A natural setting to consider transductions is that of graph classes. We say that a graph class $\Dd$ is {\em{transducible}} from a graph class $\Cc$ if there is a fixed transduction $\Tf$ such that every graph belonging to $\Dd$ can be transduced from some graph from $\Cc$ using $\Tf$. Transducibility is reflexive and transitive, and hence induces a quasi-order on graph classes. Intuitively, this quasi-order compares graph classes based on whether one graph class can be encoded in another using first-order logic.

Recently, there has been a surge of interest in transductions due to the on-going project of developing a model-theoretic structure theory for graphs, in which transducibility is the fundamental notion of embedding. A broad introduction to this rapidly developing field can be found in a recent survey of the fourth author~\cite{PilipczukSurvey}. The focus of this paper is on one particular aspect of the theory: the study of the transducibility quasi-order itself.

While the landscape painted by the transducibility quasi-order is expected to be very varied and detailed (see \cite[Figure~2]{BraunfeldNOS22transductions}), at this point we have only limited tools for proving negative results --- that some graph classes cannot be transduced from others. One technique is to exploit known {\em{transduction ideals}}: properties of graph classes closed under transductions. For instance, it is known that any class transducible from a class of bounded cliquewidth again has bounded cliquewidth, hence the class of planar graphs (which has unbounded cliquewidth) cannot be transduced from the class of graphs of cliquewidth at most $t$, for any $t\in \N$. Other transduction ideals useful for such arguments include classes of bounded twin-width~\cite{BonnetKTW22}; classes of bounded linear cliquewidth; classes of shrubdepth at most $d$, for any fixed $d\in \N$~\cite{GanianHNOM19}; or monadically stable and monadically dependent classes (see~\cite[Section 4.1]{PilipczukSurvey} for an~introduction).

However, reasoning involving known transduction ideals has a very limited strength and cannot explain the whole picture. A more detailed look has recently been offered by Braunfeld, Ne\v{s}et\v{r}il, Ossona de Mendez, and Siebertz~\cite{BraunfeldNOS22transductions}. Among other results, they characterized classes transducible from the class of trees, and they proved that the hierarchy of graph classes of bounded pathwidth is strict: for every $t\in \N$, the class of graphs of pathwidth at most $t+1$ is not transducible from the class of graphs of pathwidth at most $t$. However, the analogous question for treewidth was left open. Very recently, the third author with Hlin\v{e}n\'y~\cite{HlinenyJ25}, and independently the first and the fourth authors with Pokr\'yvka~\cite{GajarskyPP25}, proved that the class of three-dimensional grids is not transducible from the class of planar graphs (or more generally, from the class of graphs embeddable in $\Sigma$, for any fixed surface $\Sigma$). Both proofs rely on finding a certain invariant inspired by the product structure of planar graphs~\cite{DujmovicJMMUW20} that is preserved by transductions.

\paragraph*{Our results.} In this paper, we further explore the transducibility quasi-order on classes of sparse graphs. Here, by classes of sparse graphs we mean graph classes that fall within the framework of the sparsity theory introduced by Ne\v{s}et\v{r}il and Ossona de Mendez~\cite{sparsity}, most notably graph classes of bounded expansion. 
First, we prove (\cref{thm:trans-shallow}) that if $\Dd$ is a weakly sparse graph class (that is, there exists $t$ such that no member
of $\Dd$ contains $K_{t,t}$ as a subgraph) transducible from a class of bounded expansion $\Cc$, then for some $k\in \N$, every graph from $\Dd$ can be found as a $k$-congested depth-$k$ minor of a graph from $\Cc$ augmented by adding a universal vertex. Up to a minor technical aspect in the definition of a transduction, this is in fact a characterization of the transducibility quasi-order among graph classes of bounded expansion (see \cref{lem:converse} for the reverse implication).

Based on the result above, we propose a new transduction invariant: the asymptotic growth of the weak $d$-coloring number of the class, expressed as a function of $d$. Here, weak coloring numbers are a family of graph parameters that are commonly used to work with classes of sparse graphs; see the recent survey of Siebertz~\cite{SiebertzSurvey} for an introduction.
We prove (\cref{cor:trans-dominates}) that if a class $\Dd$ is transducible from a class $\Cc$, and both $\Cc$ and $\Dd$ have bounded expansion, then the weak $d$-coloring numbers of $\Dd$ cannot grow asymptotically faster with $d$ than the weak $d$-coloring numbers of~$\Cc$.

The new invariant allows us to derive a host of non-transducibility results for well-studied classes of sparse graphs, by exploiting known upper and lower bounds on their weak coloring numbers. Below we give three concrete examples of such corollaries. 

It is known that for the class of graphs of treewidth at most $t$, the weak $d$-coloring number is equal to~$\binom{d+t}{t}$~\cite{GroheKRSS18}, which is a polynomial in $d$ of degree exactly $t$. This allows us to answer the question about the strictness of the hierarchy of classes of bounded treewidth, which was left open by Braunfeld et al.~\cite{BraunfeldNOS22transductions}.

\begin{restatable}{theorem}{mainTw}\label{thm:main-tw}
	For every $t\in \N$, the class of graphs of treewidth at most $t+1$ is not transducible from the class of graphs of treewidth at most $t$.
\end{restatable}

Further, recall that the {\em{Hadwiger number}} of a graph $G$ is the largest $t$ such that $G$ contains the complete graph $K_t$ as a minor. It is known that the weak $d$-coloring number of the class of graphs of Hadwiger number at most $t$ is upper bounded by $\Oh(d^t)$~\cite{HeuvelMQRS17} and lower bounded by $\Omega(d^{t-1})$~\cite{GroheKRSS18}. Hence, we have:

\begin{restatable}{theorem}{mainHadw}\label{thm:main-hadw}
	For every positive $t\in \N$, the class of graphs of Hadwiger number at most $t+2$ is not transducible from the class of graphs of Hadwiger number at most $t$.
\end{restatable}

Finally, from the known upper bound of $\Oh(d^3)$ on the weak $d$-coloring number of planar graphs~\cite{HeuvelMQRS17}, combined with the lower bound of $\Omega(d^4)$ on the weak $d$-coloring number of graphs of treewidth at most~$4$~\cite{GroheKRSS18}, we get the following:

\begin{restatable}{theorem}{mainPlanarTw}\label{thm:main-planar}
	The class of graphs of treewidth at most $4$ is not transducible from the class of planar graphs.
\end{restatable}

\cref{thm:main-planar} stands in contrast with the result of Braunfeld et al.~\cite{BraunfeldNOS22transductions} that for every $t\in \N$, the class of graphs of pathwidth at most $t$ is in fact transducible from the class of planar graphs.
Further, while graphs of treewidth $2$ are known to be planar, \cref{thm:main-planar} leaves it open whether graphs of treewidth $3$ are transducible from planar graphs. Note that a negative answer to this question would follow from an improvement of the $\Oh(d^3)$ upper bound on the weak $d$-coloring number of planar graphs~\cite{HeuvelMQRS17}.

%
%

\section{Preliminaries}\label{sec:prelims}

By $\N$ we denote the set of nonnegative integers.

\paragraph*{Graphs.}
All graphs considered in this paper are finite, undirected and simple, and are considered up to isomorphism.
The vertex set and the edge set of a graph $G$ are denoted by $V(G)$ and $E(G)$, respectively. For a vertex subset $A\subseteq V(G)$, the induced subgraph $G[A]$ consists of all the vertices of $A$ and edges of $G$ with both endpoints in $A$. The {\em{radius}} of a connected graph $G$ is the smallest $d$ satisfying the following: there is a vertex $u$ such that every other vertex can be reached from $u$ by a path of length at most $d$.

We say that a graph class $\Cc$ is {\em{weakly sparse}} if there is $t\in \N$ such that no member of $\Cc$ contains the complete bipartite graph $K_{t,t}$ as a subgraph.

\paragraph*{Congested shallow minors and bounded expansion.} Let $G$ be a graph. We say that a vertex subset $A\subseteq V(G)$ is {\em{connected}} if the induced subgraph $G[A]$ is connected. Further, two vertex subsets $A,B\subseteq V(G)$ {\em{touch}} if they share a vertex ($A\cap B\neq \emptyset$) or are adjacent (there exist adjacent $a\in A$ and $b\in B$).

For two graphs $H$ and $G$, a {\em{model}} of $H$ in $G$ is a mapping $\eta$ that assigns each vertex $u\in V(H)$ a connected vertex subset $\eta(u)\subseteq V(G)$, called the {\em{branch set}} of $u$, so that the following condition is satisfied: whenever $u$ and $v$ are adjacent in $H$, the branch sets $\eta(u)$ and $\eta(v)$ touch. Note that we do not require the branch sets to be disjoint. For $c,d\in \N$, we say that $H$ is a {\em{congestion-$c$ depth-$d$ minor}} of $G$ if there is a model $\eta$ of $H$ in $G$ such that
\begin{itemize}
	\item for every $u\in V(H)$, the graph $G[\eta(u)]$ has radius at most $d$; and
	\item for every $v\in V(G)$, there are at most $c$ vertices $u\in V(H)$ such that $v\in \eta(u)$.
\end{itemize}
By $\Minors^c_d(G)$ we denote the set of all congestion-$c$ depth-$d$ minors of $G$. For a graph class $\Cc$, we define
\[\Minors^c_d(\Cc)\coloneqq \bigcup_{G\in \Cc} \Minors^c_d(G).\]

Note that if we set $c=1$ in the definition above, we require the branch sets to be pairwise disjoint. In this case we may speak simply about {\em{depth-$d$ minors}}, whose set we denote by $\Minors_d(G)$ (for a graph $G$) and $\Minors_d(\Cc)$ (for a class $\Cc$). The following definition is a fundamental notion in the theory of Sparsity.

\begin{definition}
	A graph class $\Cc$ has {\em{bounded expansion}} if for every $d\in \N$ there exists $p_d\in \N$ such that
	\[|E(H)|\leq p_d\cdot |V(H)|\qquad\textrm{for every }H\in \Minors_d(\Cc).\]
\end{definition}

It turns out that classes of bounded expansion are closed under taking congested shallow minors, in the following sense.

\begin{theorem}[cf. {\cite[Proposition~4.6]{sparsity}} and {\cite[Chapter 1, Corollary 2.28]{sparsityNotes}}]\label{prop:cong-minors-be}
	For all $c,d\in \N$ and every graph class $\Cc$ of bounded expansion, the class $\Minors^c_d(\Cc)$ also has bounded expansion.
\end{theorem}

Finally, $H$ is a {\em{minor}} of $G$ if there is a congestion-$1$ depth-$\infty$ model of $H$ in $G$, that is, one where the branch sets have to be disjoint and we only require them to induce connected subgraphs.

\paragraph*{Weak coloring numbers.} An {\em{ordered graph}} is a graph $G$ equipped with a total order $\preceq$ on the vertex set of $G$. Let $(G,\preceq)$ be an ordered graph. For vertices $u,v\in V(G)$ with $u\preceq v$, we say that $u$ is {\em{weakly $d$-reachable}} from $v$ if in $G$ there is a path with endpoints $u$ and $v$ and of length at most $d$ such that every vertex of this path is not smaller than $u$ in $\preceq$. For $v\in V(G)$, we define the {\em{weak $d$-reachability set}} of $v$ as
\[\WReach^{G,\preceq}_d[v]\coloneqq \{u\colon \textrm{$u$ is weakly $d$-reachable from $v$ in $(G,\preceq)$}\}.\]
We may omit the superscript in notation if the ordered graph $(G,\preceq)$ is clear from the context.

The {\em{weak $d$-coloring number}} of an ordered graph $(G,\preceq)$ is the maximum size of a weak $d$-reachability~set:
\[\wcol_d(G,\preceq)\coloneqq \max_{v\in V(G)}|\WReach^{G,\preceq}_d[v]|\]
The weak $d$-coloring number of a graph $G$ is the minimum possible weak $d$-coloring number of its ordering:
\[\wcol_d(G)\coloneqq \min_{\preceq\colon \textrm{total order on }V(G)} \wcol_d(G,\preceq).\]
Finally, we may apply this definition to any graph class $\Cc$ by taking the supremum.
\[\wcol_d(\Cc)\coloneqq \sup_{G\in \Cc} \wcol_d(G).\]
Note that $\wcol_d(\Cc)=\infty$ when the weak $d$-coloring number is unbounded on the class $\Cc$. It turns out that the boundedness of all the weak coloring numbers characterizes bounded~expansion.

\begin{theorem}[\cite{Zhu09}]\label{prop:wcols-be}
	A graph class $\Cc$ has bounded expansion if and only if $\wcol_d(\Cc)$ is finite for every $d\in \N$.
\end{theorem}

\paragraph*{First-order logic and transductions.} We assume the reader is familiar with the standard first-order logic on graphs. Below we recall standard definitions and facts about transductions; see the recent survey~\cite{PilipczukSurvey} for a more thorough introduction. 

We will work with {\em{vertex-colored graphs}}. For a set of colors $\Cols$, by a {\em{$\Cols$-colored graph}} we mean an undirected graph $G$ with a vertex subset $U_C\subseteq V(G)$ distinguished, for each color $C\in \Cols$. Note that a vertex may belong to several such subsets $U_C$, or to no subset $U_C$ at all.  We model a $\Cols$-colored graph as a relational structure whose universe is the vertex set, there is one binary relation signifying adjacency, and one unary relation per color $C\in \Cols$ that selects all the vertices that belong to $U_C$. The {\em{signature}} of a vertex-colored graph $G$ is the signature of this relational structure, i.e., consisting of one binary relation and $|\Cols|$ unary relations. This allows us to employ standard first-order logic on $\Cols$-colored graphs, where in atomic formulas one can check equality, adjacency, and color membership.

A first-order formula $\varphi(x,y)$ over $\Cols$-colored graphs is {\em{symmetric}} if $G\models \varphi(u,v)\Leftrightarrow G \models \varphi(v,u)$ for every $\Cols$-colored graph $G$ and all $u,v \in V(G)$. 
For such a formula $\varphi$ and a $\Cols$-colored graph $G$, we define the (uncolored) graph $\varphi(G)$ as follows: the vertex set of $\varphi(G)$ is the same as that of $G$, and two distinct vertices $u,v$ are adjacent in $\varphi(G)$ if and only if $\varphi(u,v)$ holds in $G$.

For an uncolored graph $G$, a \emph{$\Cols$-expansion} of $G$
is a $\Cols$-colored graph $G^+$ whose underlying uncolored
graph is $G$. Thus, $G^+$ is obtained by specifying a subset
$U_C\subseteq V(G)$ for each $C\in\Cols$, without changing
the vertex set or the edge set.

A {\em{(first-order) transduction}} $\Tf$ consists of a finite set of colors $\Cols$ and a symmetric first-order formula $\varphi(x,y)$ over the signature of $\Cols$-colored graphs. In this paper we consider only first-order transductions, hence we call them simply transductions for brevity. 
For a transduction $\Tf$ and an uncolored graph $G$, we define
\[
\Tf(G)\coloneqq
\{H : H \text{ is an induced subgraph of } \varphi(G^+)
\text{ for some $\Cols$-expansion $G^+$ of $G$}\}.
\]
For a transduction $\Tf$ and a graph class $\Cc$, define 
\[\Tf(\Cc)\coloneqq \bigcup_{G\in \Cc}\Tf(G).\]
Finally, we say that a graph class $\Dd$ is {\em{transducible}} from $\Cc$ if $\Dd\subseteq \Tf(\Cc)$ for some transduction $\Tf$. Following~\cite{PilipczukSurvey}, the intuition is as follows: $\Dd$ is transducible from $\Cc$ if every member of $\Dd$ can be encoded in a colored graph from $\Cc$ using a fixed, first-order mechanism expressed by $\varphi$.
Transducibility is transitive (cf. \cite[Lemma~2]{PilipczukSurvey}). Together with reflexivity, this makes transducibility a quasi-order on graph classes.

We remark that often in the literature, one considers a slightly stronger notion of {\em{transductions with copying}}. Technical aspects differ, but for the purpose of this paper we may adopt the following definition: a graph class $\Dd$ is {\em{transducible}} with copying from a graph class $\Cc$ if there exists a transduction $\Tf$ (without copying) and a number $p\in \N$ such that $\Dd\subseteq \Tf(\Cc\bullet p)$, where $\Cc\bullet p\coloneqq \{G\bullet p\colon G\in \Cc\}$ and $G\bullet p$ is the {\em{$p$-blowup}} of $G$: the graph obtained from $G$ by replacing every vertex with a clique consisting of $p$ copies (and copies of adjacent vertices are again adjacent). In other words, the transduction may start not with a graph from $\Cc$, but with its $p$-blowup, for some constant $p$; intuitively, this allows interpreting multiple (up to $p$) vertices of the transduced graph $H$ in a single vertex of the source graph $G$.
It can be easily seen that $G\bullet p$ is intertransducible with the graph obtained from $G$ by adding $p-1$ pendants (degree-$1$ vertices) to every vertex of $G$; that is, each of these graphs can be transduced from the other by a fixed transduction, and vice versa. Therefore, if a graph class $\Cc$ is closed under addition of degree-$1$ vertices, the notions of being transducible from $\Cc$ with copying and without copying do coincide. This will be the case for all the concrete graph classes considered in this paper, hence we will stick to the simpler notion of non-copying transductions. See also~\cite[Section~8]{BraunfeldNOS22transductions} for a more elaborate treatment of copying.

\section{Transductions and congested minors}

For a graph $G$, by $G^\circ$ we denote the graph obtained from $G$ by adding a {\em{universal vertex}}: a new vertex adjacent to all the other vertices. For a class $\Cc$, we denote $\Cc^{\circ}\coloneqq \{G^\circ\colon G\in \Cc\}$. 

The main goal of this section is to prove the following result, which intuitively says the following: among sparse graphs, transducibility implies containment as a congested shallow minor.

\begin{theorem}\label{thm:trans-shallow}
	Let $\Cc$ be a graph class of bounded expansion and $\Dd$ be a weakly sparse graph class transducible from $\Cc$. Then there exists $k\in \N$ such that $\Dd\subseteq \Minors^k_k(\Cc^\circ)$.
\end{theorem}

A few remarks are in order. First, since the property of having bounded expansion is preserved both by the operations of adding a universal vertex (easy from the definition) and of taking a congested shallow minor (\cref{prop:cong-minors-be}), from \cref{thm:trans-shallow} it follows that every weakly sparse class transducible from a class of bounded expansion also has bounded expansion. This fact was, however, already known, as it easily follows from the results of~\cite{GajarskyKNMPST20}.

Second, one might wonder whether the addition of a universal vertex is necessary in \cref{thm:trans-shallow}. It is: the class of stars is transducible from the class of edgeless graphs, but is not contained in their $k$-congested depth-$k$ minors. It turns out that adding a single universal vertex suffices to solve this issue.

Third, the converse of \cref{thm:trans-shallow} is also true if we allow copying: for every $k\in \N$ and a graph class $\Cc$ of bounded expansion, the class $\Minors^k_k(\Cc^\circ)$ is transducible with copying from $\Cc$. Since the statement of \cref{thm:trans-shallow} can be easily lifted to transductions with copying as well, this provides a purely combinatorial characterization of transducibility with copying among graph classes of bounded expansion. We postpone this discussion to the end of this section, and for now we focus on the proof of \cref{thm:trans-shallow}.


\medskip

The main tool that will be used in the proof of \cref{thm:trans-shallow} is the Local Feferman-Vaught Theorem, proposed in~\cite{Dreier23,PilipczukST18a}. We will use the formulation below that is adjusted to our terminology. This formulation follows easily from the more descriptive ones presented in~\cite{Dreier23,PilipczukST18a}. 

We say that two vertices $u,v$ of a graph $G$ are {\em{$d$-separated}} by a vertex set $S$ if every path of length at most $d$ that connects $u$ and $v$ must necessarily intersect $S$. Note that this definition also applies to the case when $u$ or $v$ belongs to $S$, in which case the condition is always satisfied.


\begin{theorem}[see {\cite[Theorem 1.12]{Dreier23}} and {\cite[Lemma~3.1]{PilipczukST18a}}]\label{thm:local-FV}
	Let $s\in \N$, $\Cols$ be a finite set of colors, and $\phi(x,y)$ be a symmetric first-order formula over the signature of $\Cols$-colored graphs. Then there exists a constant $d\in \N$, depending only on the quantifier rank of $\phi$, and a finite set of colors $\Lambda$, depending on $s$, $\Cols$, and $\phi$, such that the following holds. For every $\Cols$-colored graph $G$ and a vertex subset $S\subseteq V(G)$ with $|S|\leq s$, there is a coloring $\lambda\colon V(G)\to \Lambda$ such that for every pair of vertices $u,v\in V(G)$ that are $d$-separated by $S$, whether $\phi(u,v)$ holds in $G$ depends only on the pair of colors $(\lambda(u),\lambda(v))$. In other words, there is a set $R\subseteq \Lambda\times \Lambda$ such that for all $u,v\in V(G)$ that are $d$-separated by~$S$, we have
	\[G\models \phi(u,v)\qquad\textrm{if and only if}\qquad (\lambda(u),\lambda(v))\in R.\]
\end{theorem}

Let us now formulate the key technical lemma towards the proof of \cref{thm:trans-shallow}.


\begin{lemma}\label{lem:cong-on-wcols}
	Let $s \in \N$, $\Cols$ be a finite set of colors, $\phi(x,y)$ be a symmetric first-order formula over the signature of $\Cols$-colored graphs. Let $d\in \N$ be the constant given by \cref{thm:local-FV} for $s$, $\Cols$ and $\phi$. Next, let $G$ be a $\Cols$-colored graph and $\preceq$ be a total order on $V(G)$ so that the following two conditions are satisfied: $\wcol_{2d}(G,\preceq)\leq s$ and we have $\WReach^{G,\preceq}_d[u]\cap \WReach^{G,\preceq}_d[v]\neq \emptyset$ for all $u,v\in V(G)$. Finally, suppose $\phi(G)$ does not contain the complete bipartite graph $K_{t,t}$ as a subgraph, for some $t\in \N$. Then there exists $k\in \N$, depending only on $\Cols$, $\phi$, $d$, $s$, and $t$, such that $\phi(G)$ is a $k$-congested depth-$2d$ minor of $G$. 
\end{lemma}

Before we prove \cref{lem:cong-on-wcols}, let us derive \cref{thm:trans-shallow} from it.

\begin{proof}[Proof of \cref{thm:trans-shallow} using \cref{lem:cong-on-wcols}]
	Adding a universal vertex to a graph can increase the numbers $p_d$ from the definition of bounded expansion by at most $1$, hence the class $\Cc^\circ$ also has bounded expansion. Moreover, clearly $\Cc$ is transducible from $\Cc^\circ$, by a transduction that simply drops the added universal vertex. Hence, by transitivity of transducibility, we may assume that $\Dd$ is transduced from $\Cc^\circ$ by a transduction $\Tf$. Our goal is to prove that $\Dd\subseteq \Minors^k_k(\Cc^\circ)$, for some $k\in \N$.  As $\Dd$ is weakly sparse, there exists $t\in \N$ such that no graph in $\Dd$ contains $K_{t,t}$ as a subgraph.
	
	Let $\Cols_0$ and $\psi(x,y)$ be the set of colors and the formula used by transduction $\Tf$. Choose a fresh color $A\notin\Cols_0$, and define
	\[\Cols\coloneqq\Cols_0\cup\{A\}\qquad\textrm{and}\qquad \phi(x,y)\coloneqq\psi(x,y)\wedge A(x)\wedge A(y).\]

	For each graph $H\in\Dd$, we fix a graph $G_H\in\Cc^\circ$ and a $\Cols_0$-expansion $G_H^0$ of $G_H$ such that $H$ is an induced subgraph of $\psi(G_H^0)$. Let $G_H^+$ be the $\Cols$-expansion of $G_H$ obtained from $G_H^0$ by interpreting the new color $A$ as $V(H)$. Then $\phi(G_H^+)$ consists of $H$ together with possibly some isolated vertices. Since $H$ does not contain $K_{t,t}$ as a subgraph, neither does $\phi(G_H^+)$.

	Let $d\in \N$ be the constant provided by \cref{thm:local-FV} for $\phi$. We set
	\[s\coloneqq \wcol_{2d}(\Cc)+1,\]
	which is finite by \cref{prop:wcols-be}, because $\Cc$ is assumed to have bounded expansion.

	Consider now any graph $H\in \Dd$ and the corresponding graph $G_H\in \Cc^\circ$. Let $v$ be a universal vertex of $G_H$. Then there is a total order $\preceq'$ on the vertices of $G_H - v$ whose weak $2d$-coloring number is at most $\wcol_{2d}(\Cc)$. 
	We create a total order $\preceq$ on $V(G_H)$ by appending $v$ to the beginning of $\preceq'$ (thus making $v$ the minimum of $\preceq$).  For $\preceq$ the weak $2d$-coloring number is bounded by $\wcol_{2d}(\Cc)+1=s$. Moreover, in $\preceq$ every pair of weak $d$-reachability sets intersects, since $v \in \WReach_d[u]$ for every $u \in V(G_H)$.
	
	Since $\phi(G_H^+)$ does not contain $K_{t,t}$ as a subgraph, we may now apply \cref{lem:cong-on-wcols} to $G_H^+$ ordered by $\preceq$ to conclude that $\varphi(G_H^+)$ is a $k$-congested depth-$k$ minor of $G_H^+$, for some $k\in \N$ that depends on $\Cols,\phi,d,s,t$ (with $\Cols$ and $\phi$ depending only on $\Tf$, $d$ and $s$ depending on $\Cc$ and $\Tf$, and $t$ depending on $\Dd$). Since $G_H^+$ is a $\Cols$-expansion of $G_H$ and $H$ is an induced subgraph of $\varphi(G_H^+)$, we conclude that every $H\in \Dd$ is a $k$-congested depth-$k$ minor of the corresponding $G_H\in \Cc^\circ$.
\end{proof}

We now proceed to the proof of \cref{lem:cong-on-wcols}. In its proof we will use the following combinatorial observation.
\begin{lemma}\label{cl:Bollobas}
	Let $n,a,b \in \N$.
	Let $(A_1, B_1),\ldots,(A_n, B_n)$ be a sequence of pairs of sets such that
	\begin{itemize}
		\item $|A_i|\leq a$, $|B_i|\leq b$, and $A_i\cap B_i=\emptyset$ for all $1\leq i\leq n$, and
		\item $A_i\cap B_j\neq \emptyset$ for all $1\leq i<j\leq n$.
	\end{itemize}
Then $n\leq b^0+b^1+\ldots+b^a$.
\end{lemma}
\begin{proof}
	We first handle the case when $b=0$, with $a$ arbitrary. The lemma stipulates that we have to have $n \le 1$. We prove that having $n > 1$ leads to a contradiction: Since $b=0$, we have that $|B_2| = 0$, and so we have $A_1 \cap B_2 = \emptyset$, which contradicts the second condition in the assumptions of the lemma. 

	To prove the lemma for $b\not=0$, fix an arbitrary $b \ge 1$ in the rest of the proof. We proceed by induction on $a$. The case $a=0$ is similar to the case $b=0$: If $n> 1$, we must have $A_1 = \emptyset$ and  then $B_2 \cap A_1 = \emptyset$ contradicts the assumptions. Hence we have $n \le 1$ in this case, as desired.
	
	For the induction step, suppose for contradiction that $n> b^0+b^1+\ldots+b^a$. Note that $B_n$ intersects all the sets $A_1,\ldots,A_{n-1}$. Since $|B_n|\leq b$, there is an element $x\in B_n$ that belongs to at least $\frac{n-1}{b}$ sets $A_i$ from $\{A_1,\ldots,A_{n-1}\}$. 
	We now consider the subsequence of $(A_1,B_1), \ldots, (A_n,B_n)$ consisting of those pairs $(A_i, B_i)$ for which $x \in A_i$; we denote the pairs in this sequence as $(A_1', B_1'),\ldots,  (A_m', B_m')$. Note that by the choice of $x$ we have that $m \ge \frac{n-1}{b}> b^0+b^1+\ldots +b^{a-1}$.
	Note also that $x$ does not belong to any of the sets $B_i'$, because by the first condition in the assumptions of the lemma each $B_i'$ is disjoint from $A_i'$. Now consider the sequence $(A_1'\setminus \{x\}, B_1'),\ldots,  (A_m'\setminus \{x\}, B_m')$. This is a sequence of pairs of sets that satisfies the hypothesis for $a-1$ and $b$, so the induction hypothesis applies. The induction hypothesis tells us that the length of the sequence is at most $b^0+b^1+\ldots+b^{a-1}$, which is a contradiction with $m > b^0+b^1+\ldots+b^{a-1}$.
\end{proof}

\begin{proof}[Proof of \cref{lem:cong-on-wcols}]	
Let us start with some notation.  We fix the graph $G$ and the order $\preceq$ of $V(G)$ for the rest of the proof, hence we will omit them from superscripts in the notation. 
We denote $H\coloneqq \phi(G)$ for brevity and we let $\Lambda$ be the finite set of colors provided by \cref{thm:local-FV} for $\Cols$, $\phi$, and $s$. We may assume without loss of generality that $s\geq 2$.

For later use, we define the inverse weak $d$-reachability set of a vertex $v\in V(G)$ as follows:
\[\WReach^{-1}_d[v]\coloneqq \{u\in V(G)~|~v\in \WReach_d[u]\}.\]

We now proceed with the proof. Let us first note the following standard claim:
\begin{claim}\label{cl:local-sep}
	For every pair of vertices $u,v\in V(G)$, $u$ and $v$ are $d$-separated in $G$ by $\WReach_d[u]\cap \WReach_d[v]$.
\end{claim}
\begin{claimproof}
	Let $P$ be any path of length at most $d$ connecting $u$ and $v$, and let $w$ be the $\preceq$-minimum vertex lying on $P$. Then the subpath of $P$ between $u$ and $w$ witnesses that $w\in \WReach_d[u]$, and the subpath between $v$ and $w$ witnesses that $w\in \WReach_d[v]$. Thus $w\in \WReach_d[u]\cap \WReach_d[v]$ and therefore each path of length at most $d$ connecting $u$ and $v$ intersects $\WReach_d[u]\cap \WReach_d[v]$. This establishes the claim.
\end{claimproof}

By the assumptions of the lemma, for each $u,v \in V(G)$ we have that $\WReach_d[u]\cap \WReach_d[v]$ is nonempty, and therefore it makes sense to speak of the maximum of this set with respect to $\preceq$. The following claim is also standard.

\begin{claim}\label{cl:intersect}
	Let $u,v \in V(G)$ and let $w$ be the maximum element of $\WReach_d[u]\cap \WReach_d[v]$ with respect to $\preceq$.
	Then 
	$\WReach_d[u]\cap \WReach_d[v] \subseteq \WReach_{2d}[w]$.
\end{claim}
\begin{claimproof}
	Let $x \in \WReach_d[u]\cap \WReach_d[v]$ be arbitrary. The concatenation of the path (of length at most $d$) witnessing that $w \in \WReach_d[u]$ and the path (of length at most $d$) witnessing that $x \in \WReach_d[u]$ yields a walk of length at most $2d$ such that no vertex on the walk is smaller than $x$ with respect to $\preceq$.
	This walk can be turned into a path of length at most $2d$ witnessing that $x \in \WReach_{2d}[w]$ in the standard way.
\end{claimproof}

We now define, for each $w \in V(G)$, a subset $F(w)$ of $E(H)$ by setting
$$ F(w) := \{ uv \in E(H)~:~\text{$w$ is the $\preceq$-maximum of $\WReach_d[u]\cap \WReach_d[v]$} \}. $$
Note that if $uv \in F(w)$, then both $u$ and $v$ belong to $\WReach_d^{-1}[w]$.




We will now work towards the following goal.

\begin{claim}\label{cl:vertex-cover}
	For every vertex $w$ of $G$, there exists a set $X_w\subseteq \WReach^{-1}_d[w]$ such that
	\begin{itemize}
		\item $|X_w|\leq f(s,t,|\Lambda|)$ for some function $f\colon \N^3\to \N$, and
		\item for every edge $uv$ of $H$ with  $uv\in F(w)$, at least one endpoint of $uv$ belongs to $X_w$.
	\end{itemize}
\end{claim}

In other words, we postulate that for every $w\in V(G)$, the edges in $F(w)$ admit a vertex cover of size bounded in terms of $s$, $t$, $|\Lambda|$. In fact, we shall set
\[f(s,t,|\Lambda|)\coloneqq 2^s\cdot |\Lambda|^2\cdot \Ramsey(2t,s^s),\]
where $\Ramsey(\cdot,\cdot)$ is the usual two-colored Ramsey function: every graph on more than $\Ramsey(a,b)$ vertices contains a clique of size $a$ or an independent set of size $b$.

\begin{claimproof}
Let $M$ be an inclusion-wise maximal matching (a set of edges pairwise not sharing an endpoint) that can be formed from edges in $F(w)$. Our goal will be to show that $$|M| \le 2^{s-1}\cdot |\Lambda|^2\cdot \Ramsey(2t,s^s).$$ Then setting $X_w := V(M)$ will prove the claim (here $V(M)$ is the set of endpoints of the edges of $M$), because
\begin{itemize}
	\item $|V(M)| = 2|M| \le 2^{s}\cdot |\Lambda|^2\cdot \Ramsey(2t,s^s)$,
	\item by maximality, $V(M)$ contains an endpoint of every edge in $F(w)$, and 
	\item since each vertex of $V(M)$ is an endpoint of an edge from $F(w)$, we have $V(M) \subseteq \WReach^{-1}_d[w]$.
\end{itemize}

Assume for contradiction that $|M| > 2^{s-1}\cdot |\Lambda|^2\cdot \Ramsey(2t,s^s)$. In the remainder of the proof we will show that this leads to a contradiction with the assumption that $H$ is $K_{t,t}$-free or with \cref{cl:Bollobas}, hence proving the claim.

Recall that by the assumptions of the lemma we have that  $|\WReach_{2d}[w]| \le s$. Let us denote by $\Ss_w$ the set of all subsets of $\WReach_{2d}[w]$ that contain $w$; note that $|\Ss_w| \le 2^{s-1}$.
By the definition of $F(w)$ and by \cref{cl:intersect} we have for each $uv \in F(w)$ that $\WReach_d[u]\cap \WReach_d[v] \in \Ss_w$. 
Consequently, there is $S \in \Ss_w$ and a subset $M'$ of $M$ of size more than $|\Lambda|^2\cdot \Ramsey(2t,s^s)$ such that $\WReach_d[u]\cap \WReach_d[v] = S$ for all edges $uv \in M'$.
Let us fix such $S$ for the rest of the proof.
By \cref{thm:local-FV} there is a coloring $\lambda: V(G) \to \Lambda$ and a set $R\subseteq \Lambda\times \Lambda$ such that for all $u,v\in V(G)$ that are $d$-separated by~$S$, we have
$G\models \phi(u,v)$ if and only if $(\lambda(u),\lambda(v))\in R$. Since for each edge $uv \in M'$ the vertices $u$, $v$ are $d$-separated by $S$ (by \cref{cl:local-sep}), we have that $(\lambda(u), \lambda(v)) \in R$.

Let $M''$ be obtained from $M'$ by orienting each edge in $M'$ arbitrarily, i.e. for each edge $uv \in M'$ we put exactly one of $(u,v)$ or $(v,u)$ into $M''$. 
As $|M''| = |M'| > |\Lambda|^2\cdot \Ramsey(2t,s^s)$, there is a pair $(\alpha, \beta)$ of colors such that the set $M''' := \{ (u,v) \in M''~:~ \lambda(u) = \alpha, \lambda(v) = \beta\}$ contains more than $\Ramsey(2t,s^s)$ elements. Note that by the previous paragraph and by the construction of $M'''$ we have that $(\alpha, \beta) \in R$.

Let $(u_1,v_1),(u_2,v_2),\ldots,(u_{r+1},v_{r+1})$ be distinct edges of $M'''$, where $r=\Ramsey(2t, s^s)$ and $\lambda(u_i)=\alpha$ and $\lambda(v_i)=\beta$ for all $i\in \{1,\ldots,r+1\}$. Consider an auxiliary graph $J$ on vertex set $\{1,\ldots,r+1\}$ where indices $i<j$ are adjacent if and only if $\WReach_d[u_i]\cap \WReach_d[v_j]=S$. By Ramsey's Theorem, in $J$ we may find either a clique of size $2t$ or an independent set of size $s^s$.

Suppose first we have found  a clique of size $2t$ in $J$. Then there exists a subset $I$ of indices $\{1,\ldots, r+1\}$ such that $|I| = 2t$ and for all $i,j \in I$ with $i < j$ we have that $\WReach_d[u_i]\cap \WReach_d[v_j]=S$. By \cref{cl:local-sep} we then have for each such $i,j$ that $u_i$ and $v_j$ are $d$-separated by $S$ in $G$. Since $\lambda(u_i) = \alpha$, $\lambda(v_j) = \beta$ and $(\alpha, \beta) \in R$, by \cref{thm:local-FV} we have that $u_iv_j \in E(H)$ for all $i,j \in I$ with $i < j$.
Write $I$ as $I = \{i_1 < \ldots < i_{2t}\}$.
Then all vertices $u_{i_1},\ldots, u_{i_t}$ are adjacent in $H$ to all vertices $v_{i_{t+1}},\ldots, v_{i_{2t}}$, meaning that $H$ contains $K_{t,t}$ as a subgraph; a contradiction with our assumption on $H$.


Suppose then that we have found  an independent set of size $s^s$ in $J$. Then there exists a subset $I$ of indices $\{1,\ldots, r+1\}$ such that $|I| = s^s$ and for all $i,j \in I$ with $i < j$ we have that $\WReach_d[u_i]\cap \WReach_d[v_j] \not= S$. Observe that since $\WReach_d[u_i]\cap \WReach_d[v_i] = S$ for all $i \in I$ (this was true already for all edges in $M'$, and so in particular this is true for all edges in $M'''$, and therefore also for all $(u_i,v_i)$ with $i \in I$), we have that $S \subseteq \WReach_d[u_i]$ and $S \subseteq \WReach_d[v_i]$ for all $i \in I$. Therefore, we know that $S \subseteq \WReach_d[u_i]\cap \WReach_d[v_j]$ for all $i < j$, and since $\WReach_d[u_i]\cap \WReach_d[v_j] \not= S$, it follows that we have $\WReach_d[u_i]\cap \WReach_d[v_j]\supsetneq S$ for all $i,j \in I$ with $ i < j$.


Let $(u_1',v_1'),(u_2',v_2'),\ldots,(u_{s^s}',v_{s^s}')$ be the subsequence of
$(u_1,v_1),(u_2,v_2),\ldots,(u_{r+1},v_{r+1})$
formed by edges $(u_i,v_i)$ with $i \in I$ (that is, $(u_i',v_i') := (u_{\ell}, v_{\ell})$, where $\ell$ is the $i$-th smallest index in $I$).
For each $i\in \{1,\ldots,s^s\}$ we define
\[A_i\coloneqq \WReach_d[u_i']\setminus S\qquad \textrm{and}\qquad B_i\coloneqq \WReach_d[v_i']-S.\]
Observe that by the assumptions of the considered case, the sequences $A_1,\ldots,A_{s^s}$ and $B_1,\ldots,B_{s^s}$ satisfy the premise of \cref{cl:Bollobas}: $A_i\cap B_i=\emptyset$ for all $i\in \{1,\ldots,s^s\}$ and $A_i\cap B_j\neq \emptyset$ for all $1\leq i<j\leq s^s$. 
Also, we have $|A_i|\leq s-1$ and $|B_i|\leq s-1$ for all $i$, because $w$ belongs to $S$ and all the considered weak reachability sets. So from \cref{cl:Bollobas} we may conclude that
\[s^s\leq (s-1)^0+(s-1)^1+\ldots+(s-1)^{s-1}\leq s\cdot (s-1)^{s-1}<s^s,\]
a contradiction.

Since both outcomes of applying Ramsey's Theorem to $J$ lead to a contradiction, we conclude that we must have $|M|\leq 2^{s-1}\cdot |\Lambda|^2\cdot \Ramsey(2t,s^s)$. As we argued before, this completes the proof of \cref{cl:vertex-cover}.
\end{claimproof}

\medskip

With \cref{cl:vertex-cover} established, we proceed with constructing a model $\eta$ of $H$ in $G$, with depth bounded by $2d$ and congestion bounded by  $k:= s\cdot f(s,t,|\Lambda|) + 1$.
Recall that $V(H) = V(G)$, so for every vertex $u$ (thought of as a vertex of $H$) we need to define a set $\eta(u)$ (thought of as a subset of $V(G)$) such that $G[\eta(u)]$ has radius at most $2d$ and every vertex of $G$ is in at most $k$ such sets.
To this end, for every $u\in V(H)$, we define
\[\eta(u)\coloneqq \{u\}\cup \bigcup \{\WReach_d^{-1}[w]~|~w\in V(G)\textrm{ is such that }u\in X_w\}.\]
We now verify that $\eta$ is the sought model.

\begin{claim}\label{cl:edges}
	For every edge $uv\in E(H)$, we have $u\in \eta(v)$ or $v\in \eta(u)$.
\end{claim}
\begin{claimproof}
	Let $w$ be the maximum of $\WReach_d[u]\cap \WReach_d[v]$ with respect to $\preceq$ and note that $uv\in F(w)$. By the properties of the set $X_w$ asserted by \cref{cl:vertex-cover}, we have $u\in X_w$ or $v\in X_w$. If $u\in X_w$, then $v\in \eta(u)$ by construction, and if $v\in X_w$, then $u\in \eta(v)$ by construction.
\end{claimproof}

\begin{claim}\label{cl:radius}
	For every $u\in V(G)$, the graph $G[\eta(u)]$ is connected and has radius at most $2d$.
\end{claim}
\begin{claimproof}
	First observe that for every $w\in V(G)$, the graph $G[\WReach^{-1}_d[w]]$ is connected and has radius at most $d$. Indeed, if $v\in \WReach^{-1}_d[w]$, then all the vertices of the $v$-to-$w$ path $P$ witnessing that $w\in \WReach_d[v]$ also belong to $\WReach_d^{-1}[w]$, as every suffix of $P$ witnesses weak $d$-reachability. This implies that $v$ is at distance at most $d$ from $w$ in $G[\WReach^{-1}_d[w]]$.
	
	Now, the set $\eta(u)$ consists of $u$ and the union of a collection of several other sets of the form $\WReach_d^{-1}[w]$, each inducing a connected graph of radius at most $d$. Note also that each of those sets contains $u$, because we have $u\in X_w$ and $X_w\subseteq \WReach_d^{-1}[w]$ by construction. It follows that $\eta(u)$ induces a connected graph of radius at most $2d$.
\end{claimproof}

\begin{claim}\label{cl:congestion}
	Each vertex $v\in V(G)$ belongs to at most $s\cdot f(s,t,|\Lambda|) + 1$ sets $\eta(u)$ for $u\in V(G)$, where $f$ is the function from \cref{cl:vertex-cover}. 
\end{claim}
\begin{claimproof}
	Fix $v \in V(G)$. 
	Observe that if $v\in \eta(u)$, then $u=v$ or there exists a vertex $w$ such that $u\in X_w$ and $v\in \WReach^{-1}_d[w]$; equivalently, $w\in \WReach_d[v]$. Given $v$, there are at most $|\WReach_d[v]|\leq s$ ways to choose $w$, and then at most $|X_w|\leq f(s,t,|\Lambda|)$ ways to choose $u$. All in all, for any given $v$ there are at most $s\cdot f(s,t,|\Lambda|) + 1$ candidates for vertices $u$ satisfying $v\in \eta(u)$.
\end{claimproof}
\cref{cl:edges,cl:radius,cl:congestion} show that $\eta$ is a $k$-congested depth-$2d$ model of $H$ in $G$, where $k = s\cdot f(s,t,|\Lambda|) + 1$. This concludes the proof of the lemma.
\end{proof}


\medskip

We conclude this section with the discussion of a converse of \cref{thm:trans-shallow}. 

First, note that \cref{thm:trans-shallow} can be easily strengthened by replacing the assumption that $\Dd$ is transducible from $\Cc$, with requiring only that $\Dd$ is transducible with copying from $\Cc$. Indeed, recall that this means that $\Dd$ is transducible (without copying) from $\Cc\bullet p$ for some $p\in \N$. As $\Cc\bullet p\subseteq \Minors^p_0(\Cc)$, from \cref{prop:cong-minors-be} it follows that $\Cc\bullet p$ also has bounded expansion. So we may conclude from the original statement of   \cref{thm:trans-shallow} that $\Dd\subseteq \Minors^k_k((\Cc\bullet p)^\circ)$ for some $k\in \N$. And we easily have $\Minors^k_k((\Cc\bullet p)^\circ)\subseteq \Minors^{k'}_{k'}(\Cc^\circ)$ where $k'\coloneqq kp$, implying that $\Dd\subseteq \Minors^{k'}_{k'}(\Cc^\circ)$.

Second, we note the following converse of \cref{thm:trans-shallow}.

\begin{lemma}\label{lem:converse}
	For every $k\in \N$ and a graph class $\Cc$ of bounded expansion, the class $\Minors^k_k(\Cc^\circ)$ is transducible with copying from $\Cc$.
\end{lemma}
\begin{proof}
	Let $\Cc'\coloneqq \Cc^\circ\bullet k$.
	It is easy to see that $\Cc'$ is transducible with copying from $\Cc$. Also, both adding a universal vertex and applying the $k$-blowup for a constant $k$ preserve the property of having bounded expansion (see \cref{prop:cong-minors-be}), hence $\Cc'$ also has bounded expansion. Finally, note that $\Minors^k_k(\Cc^\circ)=\Minors_k(\Cc^\circ\bullet k)=\Minors_k(\Cc')$. As transducibility is transitive, to conclude the proof it suffices to apply the following result from~\cite{BraunfeldNOS22transductions} to $\Cc'$.
	
	\begin{claim}[{\cite[Corollary 7.6]{BraunfeldNOS22transductions}}]\label{cl:shallow-minors-trans}
		For every $k\in \N$ and a graph class of bounded expansion $\Cc'$, the class $\Minors_k(\Cc')$ is transducible from $\Cc'$.
	\end{claim}

	We note that the original statement of \cite[Corollary 7.6]{BraunfeldNOS22transductions} assumes that $\Cc'$ has bounded {\em{star chromatic number}}, but this is implied by the assumption that $\Cc'$ has bounded expansion; see the discussion in \cite{BraunfeldNOS22transductions}.
\end{proof}

The discussion above justifies the following conclusion --- a characterization of transducibility with copying among classes of bounded expansion.

\begin{corollary}
	Let $\Cc$ and $\Dd$ be graph classes of bounded expansion. Then $\Dd$ is transducible with copying from $\Cc$ if and only if there exists $k\in \N$ such that $\Dd\subseteq \Minors_k^k(\Cc^\circ)$.
\end{corollary}

\section{Asymptotics of the weak coloring numbers}

In this section we use the characterization of \cref{thm:trans-shallow} to derive our non-transducibility results: \cref{thm:main-tw,thm:main-hadw}. The main idea is to use the asymptotics of the weak coloring numbers, expressed as a function of the distance parameter $d$, as a transduction invariant. We need a few definitions to speak about this in precise terms.

For functions $f,g\colon \N\to \N\cup \{\infty\}$, we shall say that $g$ {\em{dominates}} $f$ if there exists a constant $c\in \N$ such that
\[f(n)\leq c\cdot g(cn)+c\qquad\textrm{for all }n\in \N.\]
The following simple observation is crucial: if $f$ and $g$ are polynomials of degrees $a$ and $b$, respectively, and $a>b$, then $g$ does {\em{not}} dominate $f$. This is because the expression $c\cdot g(cn)+c$ is a polynomial (in $n$) of degree $b$, while $f(n)$ is a polynomial of degree $a>b$.

For a graph class $\Cc$, we define a function
$\pi_\Cc\colon \N\to \N\cup \{\infty\}$ as follows:
\[\pi_\Cc(d)\coloneqq \wcol_d(\Cc).\]
Let us observe the following.

\begin{lemma}\label{lem:cong-wcol}
	Let $\Cc,\Dd$ be graph classes such that $\Dd\subseteq \Minors^k_k(\Cc^\circ)$, for some $k\in \N$. Then $\pi_\Cc$ dominates~$\pi_\Dd$.
\end{lemma}
\begin{proof}
	Since adding a universal vertex to a graph can increase any weak $d$-coloring number by at most $1$, it suffices to bound the weak coloring numbers of $k$-congested depth-$k$ minors of a graph in terms of the weak coloring numbers of the graph itself. Precisely, we will show that if $H$ and $G$ are graphs such that $H\in \Minors^k_k(G)$, then for every $d\in \N$ we have
	\begin{equation}\label{eq:morswin}
		\wcol_d(H)\leq k\cdot \wcol_{(4k+1)d}(G).
	\end{equation}
	By applying~\eqref{eq:morswin} to every graph $H\in \Dd$ and the graph $G\in \Cc^\circ$ that contains $H$ as a $k$-congested depth-$k$ minor, we obtain that
	\[\pi_\Dd(d)=\wcol_d(\Dd)\leq k\cdot \wcol_{(4k+1)d}(\Cc^\circ)\leq k\cdot \wcol_{(4k+1)d}(\Cc)+k=k\cdot \pi_\Cc((4k+1)d)+k,\]
	thereby proving that $\pi_\Cc$ dominates $\pi_\Dd$.
	
	Towards~\eqref{eq:morswin}, let $\eta$ be a $k$-congested depth-$k$ model of $H$ in $G$, and let $\preceq$ be a total order on $V(G)$ such that $\wcol_{(4k+1)d}(G,\preceq)=\wcol_{(4k+1)d}(G)$. For every vertex $u$ of $H$, let $\gamma(u)$ be the $\preceq$-minimal vertex of $\eta(u)$. We now define a total order $\preceq^\star$ on $V(H)$ by pulling $\preceq$ through $\gamma$: whenever $\gamma(u)\prec \gamma(u')$ for some $u,u'\in V(H)$, we set $u\prec^\star u'$. The order $\preceq^\star$ between vertices $u$ with equal $\gamma(u)$ is chosen arbitrarily.
	
	We now claim that \[\wcol_d(H,\preceq^\star)\leq k\cdot \wcol_{(4k+1)d}(G,\preceq)=k\cdot \wcol_{(4k+1)d}(G);\]
	this will prove \eqref{eq:morswin}. Towards this end, suppose $u,v\in V(H)$ are such that $u\in \WReach^{H,\preceq^\star}_d[v]$. Let $Q$ be a path in $H$ witnessing this membership: $Q$ has length at most $d$, starts in $v$, ends in $u$, and all the vertices traversed by $Q$ are not smaller in $\preceq^\star$ than $u$. Consider an edge $xy$ of $Q$. Noting that $\eta(x)$ and $\eta(y)$ touch and induce connected graphs of radius at most $k$, we may find a path $P_{xy}$ in $G$ of length at most $4k+1$ that connects $\gamma(x)$ with $\gamma(y)$ and has all its vertices in $\eta(x)\cup \eta(y)$. By concatenating all the paths $P_{xy}$ for $xy\in E(Q)$ along $Q$, we obtain a walk $W$ in $G$ of length at most $(4k+1)d$ with endpoints in $\gamma(u)$ and $\gamma(v)$ such that all the vertices visited by $W$ belong to $\bigcup_{w\in V(Q)} \eta(w)$. Since $u$ is the $\preceq^\star$-smallest vertex of $V(Q)$, by the construction of $\preceq^\star$ we conclude that $\gamma(u)$ is the $\preceq$-smallest vertex visited by $W$. This means that the walk $W$ witnesses that $\gamma(u)\in \WReach_{(4k+1)d}^{G,\preceq}[\gamma(v)]$.
	
Now, for a fixed $v$ there are at most $|\WReach_{(4k+1)d}[\gamma(v)]|\leq \wcol_{(4k+1)d}(G,\preceq)$ candidates for the vertex $\gamma(u)$, since by the reasoning above, each such vertex must belong to $\WReach_{(4k+1)d}[\gamma(v)]$. As $\eta$ is a $k$-congested model, for each $z\in \WReach_{(4k+1)d}[\gamma(v)]$ there are at most $k$ vertices $u$ of $H$ with $z=\gamma(u)$. We conclude that for a fixed $v\in V(H)$ there are at most $k\cdot \wcol_{(4k+1)d}(G,\preceq)$ candidates for a vertex $u\in V(H)$ such that $u\in \WReach_d^{H,\preceq^\star}[v]$. This means that $\wcol_d(H,\preceq^\star)\leq k\cdot  \wcol_{(4k+1)d}(G,\preceq)$, as~claimed.
\end{proof}

By combining \cref{thm:trans-shallow} with \cref{lem:cong-wcol}, we may immediately conclude the following.

\begin{corollary}\label{cor:trans-dominates}
	Suppose $\Cc$ is a graph class of bounded expansion and $\Dd$ is a weakly sparse graph class that is transducible from $\Cc$. Then $\pi_\Cc$ dominates $\pi_\Dd$.
\end{corollary}

We may now use \cref{cor:trans-dominates} to conclude our main results. We recall them for convenience.

\mainTw*
\begin{proof}
	For every $t\in \N$, let $\Tt_t$ be the class of graphs of treewidth at most $t$.
	As proved in~\cite{GroheKRSS18}, we have
	\[\pi_{\Tt_t}(d)=\binom{d+t}{t},\]
	which, for a fixed $t$, is a polynomial of degree $t$ in $d$. Since $\pi_{\Tt_t}(d)$ is a polynomial of degree $t$ and $\pi_{\Tt_{t+1}}(d)$ is a polynomial of degree $t+1$, we have that $\pi_{\Tt_t}(d)$ does not dominate $\pi_{\Tt_{t+1}}(d)$. We conclude from \cref{cor:trans-dominates} that $\Tt_{t+1}$ is not transducible from $\Tt_{t}$.
\end{proof}

\mainHadw*
\begin{proof}
	For every $t\in \N$, let $\Hh_t$ be the class of graphs of Hadwiger number at most $t$. The results of \cite{GroheKRSS18} and \cite{HeuvelMQRS17} show that
	\[\pi_{\Hh_t}(d)\geq \Omega(d^{t-1})\qquad\textrm{and}\qquad \pi_{\Hh_t}(d)\leq \Oh(d^t).\]
	In particular, $\pi_{\Hh_t}$ does not dominate $\pi_{\Hh_{t+2}}$, because the former is upper-bounded by a polynomial of degree $t$ and the latter is lower bounded by a polynomial of degree $t+1$. By \cref{cor:trans-dominates}, we conclude that $\Hh_{t+2}$ is not transducible from $\Hh_t$.
\end{proof}

\mainPlanarTw*
\begin{proof}
	Let $\Tt_4$ be the class of graphs of treewidth at most $4$ and $\mathscr{P}$ be the class of planar graphs. By the results of~\cite{GroheKRSS18,HeuvelMQRS17}, we have
	\[\pi_{\Tt_4}(d)=\binom{d+4}{4}=\Theta(d^4)\qquad\textrm{and}\qquad \pi_{\mathscr{P}}(d)\leq \Oh(d^3).\]
	As before, it now follows  directly from \cref{cor:trans-dominates} that $\Tt_4$ is not transducible from $\mathscr{P}$.
\end{proof}

\bibliographystyle{plain}
\bibliography{references}

\end{document}